\RequirePackage[l2tabu,orthodox]{nag}% nag when using outdated macros
\documentclass[USenglish]{article}
\pdfoutput=1
\usepackage[a4paper, margin=1in]{geometry}

\usepackage[utf8]{inputenc}
\usepackage[T1]{fontenc}
\usepackage{microtype}
\usepackage{lmodern}

\usepackage{amsmath}
\usepackage{amssymb}

\usepackage[colorlinks=true]{hyperref}

\usepackage{authblk}

\title{Note on ``The Complexity of Counting Surjective Homomorphisms and Compactions''}
\author{Holger Dell}
\date{}
\affil{Saarland University and Cluster of Excellence, MMCI, Saarbr\"ucken, Germany,
  \texttt{hdell@mmci.uni-saarland.de}}

\newcommand{\R}{\mathbf{R}}

\usepackage[amsmath,hyperref,thmmarks]{ntheorem}

\theoremseparator{.}
\newtheorem{theorem}{Theorem}
\newtheorem{lemma}[theorem]{Lemma}

\theoremstyle{plain}
\theorembodyfont{\upshape}

\theoremstyle{nonumberplain}
\theoremheaderfont{\itshape}
\theorembodyfont{\upshape}
\theoremsymbol{\ensuremath{_\blacksquare}}
\newtheorem{proof}{Proof}

\usepackage{mathtools}

\DeclarePairedDelimiter\abs{\lvert}{\rvert}
\newcommand\cP{\#\mathrm{P}}
\newcommand\supp{\operatorname{supp}}
\newcommand\vsurj{\operatorname{VertSurj}}
\newcommand\vesurj{\operatorname{Comp}}

\newcommand\dsub{\operatorname{S\dot ub}}

\newcommand\ind{\operatorname{Ind}}
\newcommand\aut{\operatorname{Aut}}
\newcommand\graphsu{\mathcal G}
\renewcommand\hom{\operatorname{Hom}}

\begin{document}

\maketitle

\begin{abstract}
  Focke, Goldberg, and Živný~\cite{Fockeetal} prove a complexity dichotomy for the
  problem of counting surjective homomorphisms from a large input graph~$G$ without loops to a fixed graph~$H$ that may have loops.
  In this note, we give a short proof of a weaker result:
  Namely, we only prove the $\cP$-hardness of the more general problem in which~$G$ may have loops.
  Our proof is an application of a powerful framework of Lov\'asz~\cite{lovaszbook}, and it is analogous to proofs of Curticapean, Dell, and Marx~\cite{CDM17} who studied the ``dual'' problem in which the pattern graph~$G$ is small and the host graph~$H$ is the input.
  Independently, Chen~\cite{HC} used Lov\'asz's framework to prove a complexity dichotomy for counting surjective homomorphisms to fixed finite structures.
\end{abstract}

\section{Preliminaries}
Let $\graphsu$ be the set of all unlabeled, finite graphs that may have loops and multiple edges.
Let $G,H\in\graphsu$.
We denote the vertex set of~$G$ with~$V(G)$, the set of its loops with $E_1(G)$, the set of its non-loop edges with $E_2(G)$, and the set of all edges with $E(G)=E_1(G)\cup E_2(G)$.
Let $\hom(G,H)$ be the number of homomorphisms from~$G$ to~$H$.
Let $\aut(H)$ be the number of automorphisms of~$H$.
Let $\vsurj(G,H)$ be the number of vertex-surjective homomorphisms from~$G$ to~$H$ (note that this is different from the notion in \cite{CDM17,lovaszbook}, where surjectivity has to hold also for edges).
Let $\vesurj(G,H)$ be the number of ``compactions'' from~$G$ to~$H$, that is, the number of homomorphisms that are surjective on the vertices and non-loop edges of~$H$.
For a set $S\subseteq V(H)$, we denote the subgraph of~$H$ induced by the vertices of~$S$ with~$H[S]$.
Let $\ind(F,H)$ be the number of induced subgraphs of~$H$ that are isomorphic to~$F$.
Let $\dsub(F,H)$ be the number of subgraphs~$H'$ isomorphic to~$F$ that are obtained from~$H$ by deleting vertices or non-loop edges; that is, we have $V(H')=V'\subseteq V(H)$ and $E_2(H')\subseteq E_2(H[V'])$, while $E_1(H')=E_1(H[V'])$ holds.

Analogous to the setup in \cite[Section 3]{CDM17}, we view these counting functions as infinite matrices indexed by graphs~$F\in\graphsu$, which are ordered by their total size $|V(F)|+|E(F)|$.
Then $\vsurj$ and $\vesurj$ are lower triangular matrices with diagonal entries $\vsurj(H,H)=\vesurj(H,H)=\aut(H)$, and $\ind$ and $\dsub$ are upper triangular matrices with~$1$s on their diagonals.
In particular, these matrices are invertible.

\section{Previous results}
A~graph is called \emph{reflexive} if all loops are present, and a graph is called \emph{irreflexive} if it has no loops.
Let~$\mathcal F\subset\graphsu$ be the family of all graphs that are disjoint unions of irreflexive bicliques and reflexive cliques.

\begin{theorem}[Dyer \& Greenhill \cite{DyerGreenhill}]\label{dyergreenhill}
  If $H\in\mathcal F$, then $\hom(*,H)$ can be computed in polynomial time.
  Otherwise the problem is $\cP$-hard, even when the input graphs are restricted to be irreflexive.
\end{theorem}

Let $\mathcal C\subset\mathcal F$ be the family of all graphs that are disjoint unions of irreflexive stars and reflexive cliques of size at most two.

\begin{theorem}[Focke, Goldberg, and Živný~\cite{Fockeetal}]\label{fockeetal vesurj}
  If $H\in\mathcal C$, then $\vesurj(*,H)$ can be computed in polynomial time.
  Otherwise the problem is $\cP$-hard, even when the input graphs are restricted to be irreflexive.
\end{theorem}

\begin{theorem}[Focke, Goldberg, and Živný \cite{Fockeetal}]\label{fockeetal vsurj}
  If $H\in\mathcal F$, then $\vsurj(*,H)$ can be computed in polynomial time.
  Otherwise the problem is $\cP$-hard, even when the input graphs are restricted to be irreflexive.
\end{theorem}
\section{Proof of weaker versions of Theorems \ref{fockeetal vesurj} and \ref{fockeetal vsurj}}
In this section, we establish the algorithms of Theorems~\ref{fockeetal vesurj} and~\ref{fockeetal vsurj}, and prove the weaker version of the hardness claims by a reduction from Theorem~\ref{dyergreenhill}; that is, our reduction produces input graphs that may have loops.
Every homomorphism from~$G$ to~$H$ is vertex-surjective on its image under~$G$, so the following identities hold:
\begin{align}
  \hom(G,H) &= \sum_{S\subseteq V(H)} \vsurj(G,H[S])\,.
  \label{eq: hom vsurj expanded}
  \\
  \vsurj(G,H) &= \sum_{S\subseteq V(H)} (-1)^{\abs{V(H)\setminus S}}\cdot\hom(G,H[S])\,.
  \label{eq: hom vsurj inverse expanded}
\end{align}
The second equation is the inversion of the first one, and one way to obtain it is by an application of the principle of inclusion and exclusion.
(Another way is to observe that \eqref{eq: hom vsurj expanded} is equivalent to the matrix identity $\hom = \vsurj\cdot\ind$ analogously to how this was done in \cite[Section 3]{CDM17}; inverting $\ind$ yields the matrix identity $\vsurj = \hom\cdot\ind^{-1}$ that is equivalent to \eqref{eq: hom vsurj inverse expanded}).
For compactions we get similar identities from the fact the every homomorphism from~$G$ to~$H$ is a compaction to a subgraph of~$H$ obtained by deleting vertices and non-loop edges.
We obtain $\hom=\vesurj\cdot\dsub$, and we expand this equation and its inversion for convenience. For all $G,H\in\graphsu$, we have:
\begin{align}
  \hom(G,H) &
  = \sum_{F\in\graphsu} \vesurj(G,F) \cdot \dsub(F,H)\,,
  \label{eq: hom comp sub}
  \\
  \vesurj(G,H) &= \sum_{F\in\graphsu} \hom(G,F) \cdot \dsub^{-1}(F,H)\,.
  \label{eq: hom comp sub inverse}
\end{align}
Note that the sum in~\eqref{eq: hom comp sub} is indeed finite since $\dsub(F,H)\ne 0$ holds only for finitely many graphs~$F$, namely certain subgraphs of~$H$.
Since $\dsub$ is an infinite upper triangular matrix with $1$s on its diagonal, it has an inverse matrix $\dsub^{-1}$, which is also upper triangular with $1$s on its diagonal, and so $\dsub^{-1}(F,H)\ne 0$ holds only if $\dsub(F,H)\ne 0$, and the sum in~\eqref{eq: hom comp sub inverse} is also finite.

\subsection{Algorithms}
The algorithms for Theorems~\ref{fockeetal vesurj} and~\ref{fockeetal vsurj} immediately follow from~\eqref{eq: hom vsurj inverse expanded} and \eqref{eq: hom comp sub inverse} since we want to compute the left sides of the equations and can, respectively, compute the right sides in polynomial time using Theorem~\ref{dyergreenhill}.
For the case of $\vsurj$, note that deleting any vertices of $H\in\mathcal F$ again yields a graph in~$\mathcal F$.
For the case of $\vesurj$, note that deleting any vertices and non-loop edges of $H\in\mathcal C$ yields a graph~$H'\in\mathcal C$, and that $\mathcal C\subseteq\mathcal F$ holds. (Indeed, $\mathcal C$ is the unique maximal subset of~$\mathcal F$ that is closed under taking subgraphs in the sense of $\dsub$.)

\subsection{Hardness}
We use two ingredients.
The first is the following fact for the disjoint union \cite[(5.28)]{lovaszbook}.
\begin{equation}
  \hom(G\cup F,H)=\hom(G,H)\cdot\hom(F,H)
\end{equation}
The second is the following lemma proved by Lov\'asz.
\begin{lemma}[Proposition 5.43 in \cite{lovaszbook}]
  Let $S\subseteq\graphsu$ be a finite set of unlabeled graphs that is closed in the sense that, for all~$H\in S$, the set~$S$ contains all homomorphic images~$H'\in\graphsu$ of~$H$.
  Then the $(S\times S)$-matrix~$A$ with $A_{F,H}=\hom(F,H)$ is invertible.
\end{lemma}
The following lemma is completely analogous to its dual version in~\cite[Lemma~3.6]{CDM17}.
\begin{lemma}\label{lem: dual monotonicity}
  Let $\alpha:\mathcal G\to\R$ be a function of finite support and let $f:\mathcal G\to\R$ be the graph parameter with
  \begin{equation}
    f(G) = \sum_{H\in\mathcal G} \alpha(H)\cdot\hom(G,H)\,.
  \end{equation}
  When given oracle access to~$f$, we can compute $\hom(.,H)$ in polynomial time for all~$H\in\supp(\alpha)$.
\end{lemma}
\begin{proof}
  Let $T=\supp(\alpha)$ be the support of $\alpha$, that is, the set of all graphs~$H\in\graphsu$ with $\alpha(H)\ne 0$.
  Let $S\subseteq\mathcal G$ be the set of all homomorphic images of graphs in~$T$.
  For each~$F\in S$, we have:
  \begin{equation}
    f(G\cup F) = \sum_{H\in\mathcal G} \alpha(H)\cdot\hom(G,H)\cdot\hom(F,H)\,.
  \end{equation}
  We define a vector~$\beta\in\R^S$ with $\beta(H)=\alpha(H)\cdot\hom(G,H)$.
  Let $A$ be the $(S\times S)$-matrix with $A_{F,H}=\hom(F,H)$.
  By the previous lemma, this matrix is invertible.
  Finally, let $b\in\R^S$ be the vector with $b(F)=f(G\cup F)$.
  Then $b$ can be written as the matrix-vector product:
  $
    b=A\cdot\beta
  $.
  Thus we have $\beta=A^{-1}\cdot b$.
  The vector~$b$ can be computed in polynomial time by querying the oracle for the values~$f(G\cup F)$.
  The matrix~$A$ can be computed in constant time since it only depends on the fixed function~$f$.
  Thus we can compute the entire vector~$\beta$.
  In particular, for each~$H\in\supp(\alpha)$, we can determine~$\hom(G,H)$ via the identity $\hom(G,H)=\beta(H)/\alpha(H)$ since $\alpha(H)\ne 0$.
\end{proof}
Applying this lemma to $f(*)=\vsurj(*,H)$ and $f(*)=\vesurj(*,H)$ yields the hardness of Theorems~\ref{fockeetal vesurj} and~\ref{fockeetal vsurj}.
The reason is that both functions~$f$ can written as a linear combination of~$\hom(G,H)$ via \eqref{eq: hom vsurj inverse expanded} and~\eqref{eq: hom comp sub inverse}, and in both cases a counting function $\hom(*,F)$ that is hard by Theorem~\ref{dyergreenhill} appears in the support of~$\alpha$.

For Theorem~\ref{fockeetal vsurj}, note that $\vsurj(G,H)=\sum_{F} \alpha(F)\cdot\hom(G,F)$ satisfies $\alpha(H)=1$ since $S=V(H)$ is the only term in~\eqref{eq: hom vsurj expanded} where $H[S]$ is isomorphic to~$H$.
Thus if $H\not\in\mathcal F$, then $\hom(*,H)$ is $\cP$-hard by Theorem~\ref{dyergreenhill}, and so $\vsurj(*,H)$ is $\cP$-hard by Lemma~\ref{lem: dual monotonicity}.

For Theorem~\ref{fockeetal vesurj}, recall that $\dsub^{-1}(H,H)=1$ holds.
Thus if $H\not\in\mathcal F$, then $\hom(*,H)$ is $\cP$-hard by Theorem~\ref{dyergreenhill} and it reduces to $\vesurj(*,H)$ by Lemma~\ref{lem: dual monotonicity}, so the latter is hard as well.
Now suppose $H\in\mathcal F\setminus\mathcal C$.
Then there is a non-loop edge~$e$ such that $(H-e)\not\in\mathcal F$.
Clearly $\hom(*,H-e)$ is $\cP$-hard by Theorem~\ref{dyergreenhill}.
It remains to show that $\dsub^{-1}(H-e,H)\ne 0$ so that Lemma~\ref{lem: dual monotonicity} reduces $\hom(*,H-e)$ to $\vesurj(*,H)$.
Using the fact that $\dsub$ and $\dsub^{-1}$ are upper triangular matrices and that $\dsub\cdot\dsub^{-1}$ is the infinite identity matrix, we have:
\begin{align*}
  0 &= (\dsub\cdot\dsub^{-1})(H-e,H)
    =
    \underbrace{\dsub(H-e,H-e)}_{=1} \cdot \dsub^{-1}(H-e,H) + \dsub(H-e,H) \cdot \underbrace{\dsub^{-1}(H,H)}_{=1}
    \,.
\end{align*}
This implies $\dsub^{-1}(H-e,H)=-\dsub(H-e,H)\ne 0$ as required.

\paragraph{Acknowledgments.}
I thank Jacob Focke, Leslie Ann Goldberg, and Standa Živný for comments on an earlier version of this note, and for subsequent discussions at the Dagstuhl Seminar 17341 on ``Computational Counting'' in August~2017.
I thank Radu Curticapean and Marc Roth for many discussions and comments.

\bibliographystyle{plainurl}
\bibliography{references}

\end{document}